\documentclass[11pt]{article}
\usepackage[a4paper,margin=1in]{geometry}
\usepackage[english]{babel}
\usepackage{xr}

\usepackage{setspace}
\usepackage{tikz}
\usetikzlibrary{arrows,automata,positioning}
\usepackage{pgfplots}
\pgfplotsset{compat=newest}

\usepackage{mathrsfs}
\usepackage{url}
\usepackage{hyperref}
\usepackage{breakurl}
\hypersetup{breaklinks=true}
\usepackage[]{amsmath}
\usepackage{amsthm}
\usepackage{fix-cm}
\usepackage[]{amssymb}
\usepackage[]{latexsym}
\usepackage[right]{eurosym}
\usepackage[T1]{fontenc}
\usepackage[]{graphicx}
\usepackage[]{epsfig}
\usepackage{fancyhdr}
\usepackage{multirow}
\usepackage{subcaption}
\usepackage{wrapfig}

\allowdisplaybreaks

\makeatletter

\newcommand{\bbr}{\mathbb{R}}

\newcounter{modcount}
\newcommand{\modulo}[2]{%
\setcounter{modcount}{#1}\relax
\ifnum\value{modcount}<#2\relax
\else\relax
\addtocounter{modcount}{-#2}\relax
\modulo{\value{modcount}}{#2}\relax
\fi}
\newcommand{\tablepictures}[4][c]{\begin{tabular}[#1]{@{}c@{}}#2\vspace{0.5cm}\\(\alph{#4}) #3\end{tabular}}
\newcounter{gridsearch}
\newcommand{\tabpic}[2]{
    \stepcounter{gridsearch}
    \modulo{\thegridsearch}{2}
    \ifnum\value{modcount}=0
        \tablepictures[t]{#1}{#2}{gridsearch}\\[2.0cm]
    \else
        \tablepictures[t]{#1}{#2}{gridsearch}&~&
    \fi
}

\makeatother
\newtheorem{lemma}{Lemma}[section]
\newtheorem{proposition}[lemma]{Proposition}

\newtheorem{corollary}[lemma]{Corollary}
\newtheorem{definition}[lemma]{Definition}
\newtheorem{example1}[lemma]{Example}
\newtheorem{rem1}[lemma]{Remark}
\newtheorem{assumption}[lemma]{Assumption}
\newtheorem{alg1}[lemma]{Algorithm}
\newtheorem{me1}[lemma]{Mechanism}

\newenvironment{remark}{\begin{rem1}\rm}{\end{rem1}}
\newenvironment{example}{\begin{example1}\rm}{\end{example1}}

\usepackage{color}

\newcommand\ind[1]{\mathbb{I}_{\{#1\}}}

\begin{document}

\title{Clearing prices under margin calls and the short squeeze}
\author{Zachary Feinstein \thanks{Stevens Institute of Technology, School of Business, Hoboken, NJ 07030, USA, {\tt zfeinste@stevens.edu}}}
\date{\today}
\maketitle
\abstract{
In this paper, we propose a clearing model for prices in a financial market due to margin calls on short sold assets.  In doing so, we construct an explicit formulation for the prices that would result immediately following asset purchases and a margin call.  The key result of this work is the determination of a threshold short interest ratio which, if exceeded, results in the discontinuity of the clearing prices due to a feedback loop.  
~\\[1em]
}

\section{Introduction}\label{sec:intro}

In early 2021, the online community r/WallStreetBets began a targeted campaign of retail investors to purchase, or otherwise increase the price, of a few specific stocks.  Most notably, this investment campaign led to large price swings in GameStop Corp.\ (GME).  This, likewise, caused the distress of certain hedge funds such as Melvin Capital Management LP due to their large short positions in this stock.  The purpose of this work is to provide a model of prices for assets with large short positions; in particular, those positions are subject to margin calls and can face a short squeeze.

Margin calls and the short squeeze are, in some sense, the mirror image of a fire sale and traditional price-mediated contagion.  That, more traditional, setting considers the situation in which investors sell assets to satisfy regulatory requirements in a stress scenario; those asset liquidations cause the value of assets to decrease and, as such, increase the initial stress further.  This feedback effect can lead to significant price drops.  Fire sales and price-mediated contagion has been studied in, e.g.,~\cite{CFS05,AFM16,feinstein2015illiquid,BW19}.  In contrast, in this work we are focused on a setting in which price increases have feedback effects resulting in greater asset purchasing and even greater price increases.

This work is related to other recent papers studying the short squeeze episode of 2021.  Similar to the results herein, we refer the interested reader to, e.g.,~\cite{guimaraes2021short} for an economic model of a short squeeze which finds that short selling can lead to an asset bubble by causing a spike in asset demand (in order to close the position).  Empirically the 2021 short squeeze event was investigated within, e.g.,~\cite{hasso2021participated,klein2021note}.

There are two primary innovations provided by this work.  First, we provide an algorithm for finding the realized clearing prices for an asset with non-negligible number of short sales subject to margin calls.  This formulation allows for counterfactual testing of different short selling situations as well as sensitivity of the clearing solution to the various parameters of the system.  Analytical results are provided in the specific case of a linear inverse demand function.  Second, as a direct consequence of this analytical formulation, we find a threshold short interest ratio which determines whether margin calls lead to a short squeeze and a resultant discontinuity in the clearing prices.  As such, this threshold short interest ratio can be used to determine unstable market configurations.  This threshold can be utilized by investors who wish to target heavily short sold assets; similarly regulators can utilize such a result to determine short selling constraints that promote the tradeoffs between market efficiency and price stability. 

The organization of this work is as follows.  In Section~\ref{sec:margin}, we provide background information on short selling and a simple model for margin calls on such obligations.  This is extended in Section~\ref{sec:clearing} in which we determine the prices which clear a financial market with external investors and the possibility for margin calls on short sold assets.  Section~\ref{sec:squeeze} considers the implication of these clearing prices.  In particular, a threshold short interest ratio is determined below which the prices are continuous with respect to the actions of external investors but above which prices can jump.  This model is then tested against data from early 2021 in two case studies of stocks targeted by the online community r/WallStreetBets in Section~\ref{sec:cs}; a final speculative case study is presented for a stock that is heavily short sold at the time of this writing.  
The proofs of all results are provided within the appendix.

\section{Margin calls on short sales}\label{sec:margin}
Notationally, let $S > 0$ be the total shares short sold by some financial institution at the initial price of $\bar p > 0$.  On these shares, the institution has posted an initial margin of $M = (1+\alpha_0)S\bar p$ for some $\alpha_0 > 0$; that is, the short selling institution must post all proceeds from the original sale of the asset as well as an additional $\alpha_0$ proportion in cash.  Within the United States, this initial margin is set at $\alpha_0 = 50\%$ by Federal Reserve Board Regulation T.  This margin is to guarantee that the short selling institution will honor its agreement to return the $S$ assets at some future date.  For simplicity, we will assume that there is no interest charged on the short selling firm for borrowing the $S$ shares.

As prices change from $\bar p$ to $p$, the short selling institution must guarantee that they retain at least $(1+\mu)Sp$ in the margin account. This level $\mu \in (0,\alpha_0)$ is called the maintenance margin.  Within the United States, this maintenance margin is set at $\mu \geq 25\%$ (see, e.g.,~Federal Reserve Board Regulation T).  For the purposes of this work, we ignore the possibility of the firm withdrawing holdings from the margin account as the prices drop.  Therefore, the margin account is unaffected so long as $M \geq (1+\mu)Sp$.

There are two basic ways in which the short selling institution may satisfy this maintenance margin when $M < (1+\mu)Sp$:
\begin{enumerate}
\item the firm can post additional cash in the amount of $[(1+\mu)Sp - M]^+$ so that the margin account satisfies the maintenance margin; or
\item the firm returns $\Gamma$ shares of the asset so that $M \geq (1+\mu)(S-\Gamma)p$, i.e., $\Gamma := [S - \frac{M}{(1+\mu)p}]^+$ as the minimal required shares to return.
\end{enumerate}
Within this work we assume that the firm chooses to solely return shares rather than post additional cash as it has the lower cost as $\Gamma p = [Sp - \frac{M}{1+\mu}]^+ \leq [(1+\mu)Sp - M]^+$.

\begin{remark}\label{rem:margin}
The proposed clearing model presented below can be adjusted to account for the firm to both post cash and return shares simultaneously to satisfy the margin requirement.  Specifically, any posted cash can be included in the margin account $M$ which is then used to determine the resulting assets purchased and clearing prices.
\end{remark}

\section{Clearing prices}\label{sec:clearing}
As introduced in the prior section, as prices rise the margin requirements may require a firm that has short sold assets to purchase assets back.  Within this section, we are interested in the dynamics that prices take.  In particular, we are interested in how prices rise with asset purchases -- notably such constructions will not depend on which market participant is transacting.  As utilized within the fire sale literature, we will consider an inverse demand function for this purpose, i.e., if $x \geq 0$ assets are purchased in the financial market then the resulting price is $f(x)$.
\begin{assumption}\label{ass:idf}
The inverse demand function $f: \bbr_+ \to \bbr_{++}$ is a twice differentiable, strictly increasing, and concave function with $f(0) = p_0 > 0$.  Furthermore, we assume this initial price $p_0$ is high enough so that $(1+\mu)Sp_0 \leq M$, i.e., no margin call occurs without outside interventions; for notation let $\alpha = \frac{M - Sp_0}{Sp_0} \in [\mu,\infty)$, i.e., $M = (1+\alpha)Sp_0$.
\end{assumption}
\begin{example}\label{ex:linear}
Throughout this work, we will often consider a linear inverse demand function to demonstrate the specifics of the provided model.  That is, if $x \geq 0$ assets are purchased in the financial market then the resulting price is:
\[f(x) := p_0(1 + b x)\]
for market impact parameter $b > 0$.  This linear inverse demand function is prominently utilized in the fire sale literature; we refer the interested reader to, e.g.,~\cite{GLT15,CS19,feinstein2020interbank}.  
\end{example}

Before continuing, we wish to summarize notation utilized in this work.  In particular, we consider all values in physical units of assets and cash.  That is, $S$ provides the total number of shares of the asset that have been sold short.  We wish to highlight that, in practice, the number of short sold shares is often normalized by the average daily volume [ADV] $V > 0$; this ratio $s := S/V$ is often called the ``short interest ratio'' or the ``days-to-cover ratio.'' The current margin account is given by $M = (1+\alpha)Sp_0$.  Finally, the external capital purchases are provided in cash valuation and denoted by $C$.

With the external capital purchases $C > 0$, and the possibility of asset purchases to satisfy the margin requirements, the price $p$ must satisfy the clearing equation
\begin{equation}\label{eq:clearing}
p = f\left(\frac{C}{p} + \left[S - \frac{M}{(1+\mu)p}\right]^+\right).
\end{equation}
That is, the price must satisfy an equilibrium with the number of assets being purchased externally to the short selling institution ($C/p$) and those purchased to satisfy the margin requirements ($[S - \frac{M}{(1+\mu)p}]^+$).
The following proposition guarantees there exists some clearing price.  In fact, as provided in Lemma~\ref{lemma:p} below, there exists a \emph{realized} clearing price for which we can give an explicit formulation.
\begin{proposition}\label{prop:exist}
There exists some clearing price $p \in [p_0,f(C/p_0+S)]$ satisfying~\eqref{eq:clearing}.
\end{proposition}

Now, we wish to consider an explicit construction for a clearing price $p^*$ of~\eqref{eq:clearing} which is \emph{realized} in the sense that no assets are purchased that are not forced to occur.  As such, this clearing price can be constructed as the result of the fictitious margin call notion (akin to the fictitious default algorithm of~\cite{EN01,RV13}) or a t\^atonnement process (as in, e.g.,~\cite{feinstein2017currency,feinstein2020interbank}).
\begin{lemma}\label{lemma:p}
There exists a \emph{unique} clearing price $p^*$ to~\eqref{eq:clearing} satisfying the following algorithm:
\begin{enumerate}
\item determine the \emph{unique} positive price assuming no margin call is required, i.e., $p^* = f(\frac{C}{p^*})$ with $p^* > 0$; 
\item if $(1+\mu)Sp^* \leq M$ then terminate and report $p^*$;
\item if $(1+\mu)Sp^* > M$ then determine the \emph{unique} price resulting in a margin call, i.e., $p^* = f(S + \left[C - \frac{M}{1+\mu}\right]/p^*)$ with $(1+\mu)Sp^* > M$. 
\end{enumerate}
Furthermore, no margin call is required if and only if 
\begin{equation}\label{eq:C*}
C \leq \frac{(1+\alpha)p_0}{1+\mu}f^{-1}\left(\frac{(1+\alpha)p_0}{1+\mu}\right) =: C^*.
\end{equation}
\end{lemma}

The threshold $C^* > 0$ defined in~\eqref{eq:C*} provides the amount of external capital purchases necessary to trigger a margin call.  Within that formulation, $\frac{1+\alpha}{1+\mu}p_0$ is exactly the price point such that the margin call is triggered.  Notably, these thresholds do \emph{not} depend on the number of short sold shares $S$ but only on the relative amount of margin posted $\alpha$ and the maintenance margin $\mu$.

\begin{example}\label{ex:linear-p}
Consider the linear inverse demand function of Example~\ref{ex:linear}. The unique clearing price $p^*$ to~\eqref{eq:clearing} constructed from the algorithm of Lemma~\ref{lemma:p} takes value:\footnote{The algebraic details of this computation are provided within Appendix~A of the preprint version of this work available at \url{https://arxiv.org/pdf/2102.02176v3.pdf}.}
\begin{align}
\label{eq:p*-1} p^* &= p_0 \times \begin{cases} \frac{1 + \sqrt{1 + 4bC/p_0}}{2} &\text{if } C \leq C^* := \frac{p_0}{b}\frac{1+\alpha}{1+\mu}\left[\frac{\alpha-\mu}{1+\mu}\right]\\
    \frac{1 + bS + \sqrt{(1 + bS)^2 + 4b\left[C - \frac{M}{1+\mu}\right]/p_0}}{2} &\text{if } C > C^*. \end{cases} 
\end{align}
\end{example}

To conclude our discussion of the clearing prices, we wish to consider the sensitivity of these prices $p^*$ on the external capital purchases $C$ and the number of short sold shares $S$.  
In doing so, within this proposition, we only study the cases where $C \neq C^*$; that is, we only study the sensitivity of the clearing prices when (strictly) no margin call is required or when there is a margin call.  The setting in which $C = C^*$ is studied in depth within Section~\ref{sec:squeeze} below.
\begin{proposition}\label{prop:sensitivity}
Consider the clearing price $p^*$ constructed as in Lemma~\ref{lemma:p} and assume $C \neq C^*$. Then $p^*$ is nondecreasing in the external capital purchases $C$ and the number of short sold shares $S$ with sensitivities:
\begin{align*}
\frac{\partial}{\partial C} p^* &= \left.\begin{cases} \frac{f'(\frac{C}{p^*})}{p^* + \frac{C}{p^*} f'(\frac{C}{p^*})} &\text{if } C < C^* \\ 
    \frac{f'(S + \left[C - \frac{M}{1+\mu}\right]/p^*)}{p^* + \left(\left[C - \frac{M}{1+\mu}\right]/p^*\right)f'(S + \left[C - \frac{M}{1+\mu}\right]/p^*)} & \text{if } C > C^* \end{cases} \right\} > 0 \\
\frac{\partial}{\partial S} p^* &= \left. \begin{cases} 0 &\text{if } C < C^* \\ 
    \frac{\left[p^* - \frac{(1+\alpha)p_0}{1+\mu}\right] f'(S + \left[C - \frac{M}{1+\mu}\right]/p^*)}{p^* + \left(\left[C - \frac{M}{1+\mu}\right]/p^*\right)f'(S + \left[C - \frac{M}{1+\mu}\right]/p^*)} & \text{if } C > C^* \end{cases} \right\} \geq 0.
\end{align*}
\end{proposition}


\section{The short squeeze}\label{sec:squeeze}
For the purposes of this section, we will focus solely on the clearing price constructed in Lemma~\ref{lemma:p}.  
In particular, we will focus on the change in prices at the cutoff threshold for the margin call $C^*$ as provided in~\eqref{eq:C*}.
We wish to highlight that the threshold for no margin calls to occur is independent of the size of the short sale as evidenced from the form of $C^*$.

\begin{definition}\label{defn:squeeze}
The \emph{size of the short squeeze} is the difference between the price with a margin call with that without a margin call at external capital purchase of $C^*$, i.e.,
\[\delta := \lim_{C \searrow C^*} p^*(C) - \lim_{C \nearrow C^*} p^*(C)\]
with explicit dependence of the clearing price on capital purchases with $p^*$ constructed as in Lemma~\ref{lemma:p}. 
\end{definition}
\begin{remark}\label{rem:squeeze-lower}
The limiting clearing price \emph{without} a margin call at the external capital purchase of $C^*$ is equal to $\lim_{C \nearrow C^*} p^*(C) = \frac{1+\alpha}{1+\mu} p_0$ by construction of $C^*$.
\end{remark}

\begin{lemma}\label{lemma:squeeze}
The size of the short squeeze is strictly positive $\delta > 0$ if and only if $S > S^* := f^{-1}(\frac{(1+\alpha)p_0}{1+\mu}) + \frac{(1+\alpha)p_0}{(1+\mu) f'(f^{-1}(\frac{(1+\alpha)p_0}{1+\mu}))}$ physical shares of the asset have been short sold.
\end{lemma}

Before studying the size of the short squeeze under the linear inverse demand function scenario in Example~\ref{ex:linear-squeeze}, we wish to give a few comments regarding the threshold value $S^*$ proposed in Lemma~\ref{lemma:squeeze}.  
First, consider the form of $S^*$.  Recall that $\frac{1+\alpha}{1+\mu}p_0$ is the threshold price for a margin call.  Therefore $f^{-1}(\frac{1+\alpha}{1+\mu}p_0)$ provides exactly the number of shares of the asset to be purchased to reach a value of $C^*$ (i.e., to be on the cusp of triggering the margin call).  Likewise the value $f'(f^{-1}(\frac{1+\alpha}{1+\mu}p_0))$ provides the (marginal) market impact after purchasing $C^*$ worth of shares of the asset. Taken together, this threshold short sold volume $S^*$ is the sum of the number of assets purchased externally by $C^*$ (given by $f^{-1}(\frac{(1+\alpha)p_0}{1+\mu})$) and the number of shares that would need to be purchased given the ``current'' market impact in order to reach the threshold price (given by $\frac{(1+\alpha)p_0}{(1+\mu) f'(f^{-1}(\frac{(1+\alpha)p_0}{1+\mu}))}$). 

Furthermore, we find that the value $S^*$ proposed in Lemma~\ref{lemma:squeeze} is of high importance due to its relevance with financial stability.  Specifically, if the number of short sold shares $S$ is less than $S^*$ then the realized clearing price $p^*$ from Lemma~\ref{lemma:p} is continuous in the external capital purchases $C$. This means that any marginal change in the system parameters results in only a marginal change in the clearing price, i.e., market stability is achieved.  However, if the number of short sold shares exceeds $S^*$, then at $C^*$ external capital purchases market behavior becomes wholly unstable and a marginal change in system parameters can result in a drastically different clearing price (as encoded by, e.g., the size of the short squeeze $\delta$).  

From a regulatory perspective, this threshold short sales provide a tradeoff between the financial stability and market efficiency (see, e.g., \cite{bris2007efficiency,saffi2011price}).  On the extreme cases, disallowing any short selling is harmful for market efficiency and accurate price discovery (\cite{bris2007efficiency,saffi2011price}) but allowing for unlimited short selling can result in a strictly positive short squeeze if $S > S^*$.  As the short squeeze is a rare event, placing such a limit on short sales would both promote financial stability and allow for much of the market efficiency gains.

We now wish to provide an explicit formulation for the size of the short squeeze and, in particular, necessary and sufficient conditions for the existence of a short squeeze under the linear inverse demand function of Example~\ref{ex:linear}.
\begin{example}\label{ex:linear-squeeze}
Consider the linear inverse demand function of Example~\ref{ex:linear}.
The size of the short squeeze is given by $\delta = \left[bS - \frac{1 - \mu + 2\alpha}{1+\mu}\right]^+ p_0$ with $S^* = \frac{1 - \mu + 2\alpha}{b(1+\mu)}$.\footnote{The algebraic details of this computation are provided within Appendix~B of the preprint version of this work available at \url{https://arxiv.org/pdf/2102.02176v3.pdf}.}
\end{example}

To conclude this discussion of the behavior of the short squeeze, we wish consider the sensitivity of $\delta$ to the number of short sold shares $S$.  In particular, we are interested in how the size of the short squeeze changes in the number of short sold shares when $S > S^*$.
\begin{corollary}\label{cor:sensitivity}
The size of the short squeeze $\delta$ is nondecreasing in the number of sold assets $S$ with sensitivity:
\begin{align*}
\frac{\partial}{\partial S} \delta &= \frac{\left[p^* - \frac{(1+\alpha)p_0}{1+\mu}\right] f'\left(S - \frac{(1+\alpha)p_0}{(1+\mu)p^*}\left[S - f^{-1}\left(\frac{(1+\alpha)p_0}{1+\mu}\right)\right]\right)}{p^* - \frac{(1+\alpha)p_0}{(1+\mu)p^*}\left[S - f^{-1}\left(\frac{(1+\alpha)p_0}{1+\mu}\right)\right]f'\left(S - \frac{(1+\alpha)p_0}{(1+\mu)p^*}\left[S - f^{-1}\left(\frac{(1+\alpha)p_0}{1+\mu}\right)\right]\right)} \ind{S > S^*} \geq 0. 
\end{align*}
\end{corollary}

\section{Case Studies}\label{sec:cs}
Within this section we wish to consider three separate case studies.  The first two are to study the short squeeze events of early 2021.  The final case study is a speculative example of a heavily shorted stock at the time of this writing.
To simplify these examples we consider constant parameters $\alpha = 0.45$ and $\mu = 0.30$ to conform with regulatory requirements and provide simple heuristic values.  Throughout this section we will solely focus on the linear inverse demand function setting of Example~\ref{ex:linear}.  Notably, this construction provides a threshold $S^* \approx 1.23/b$ inversely proportional to market impacts for the existence of a strictly positive short squeeze.  That is, prices become unstable and jump if the shorted quantity $S$ exceeds $S^*$ so long as external investors inject more than $C^*$ into the stock.
Additionally, to simplify these considerations, we will proceed with an estimated market impact of $b = 2/V$ where $V$ denotes the average daily volume; therefore we consider these case studies with threshold $S^* \approx 0.615V$.  In this way there is, equivalently, a threshold short interest ratio $s^* \approx 0.615$ which we can utilize.  Similarly, the threshold external capital expenditures to trigger the margin call is approximately 6.43\% of the ADV, i.e., $C^* \approx 0.0643Vp_0$.\footnote{All data utilized in this section was collected from Yahoo Finance.}

\subsection{Case studies of r/WallStreetBets}\label{sec:cs-wsb}
In this section we wish to consider two stocks that have been part of a coordinated action from the online community r/WallStreetBets in early 2021.  Specifically we focus on GameStop Corp.\ (GME) and AMC Entertainment Holdings Inc.\ (AMC).  All data utilized in this section was collected so as to be timed prior to the actions of the online community r/WallStreetBets in early 2021.

\begin{example}\label{ex:gme}[GameStop Corp.]
In this first case study, we consider GME stock in early 2021.  As of December 15, 2020 -- prior to the early 2021 price movements for the stock -- there are 69.75 million shares of the stock outstanding; of those shares, 46.89 million are actually floating and available for transactions.  The ADV for all trading days in 2020 was approximately $V = 6.68$ million shares per day.  Critically for this work, as of December 15, 2020, a total of 68.13 million shares were shorted.  That is, the short interest ratio is given by $s = 68.13/6.68 \approx 10.2$.  Notably this short interest ratio is far in excess of $s^* \approx 0.615$.  As such, we find that a strictly positive short squeeze occurs with size $\delta = p_0\left(b sV - \frac{1 - \mu + 2\alpha}{1+\mu}\right) \approx 19.17p_0$, i.e., the price would jump by more than 1900\% due to a short squeeze, at an external capital purchase of $C^* \approx \$7.3$ million.  With a price per share of approximately $p_0 = \$17$/share prior to the short squeeze, we would anticipate a sudden jump in prices to over \$340/share due to the short squeeze; in fact, under such a scenario we compute the clearing price $p^* = \$344.84$/share.  This is a jump in price from approximately \$18.96/share before the short squeeze (due entirely to the external capital purchases of $C^*$).  This is consistent with price movements observed in early 2021 leading up to a spike in prices on January 27, 2021. 
\end{example}

\begin{example}\label{ex:amc}[AMC Entertainment Holdings Inc.]
In this second case study, we consider AMC stock in early 2021.  As of January 15, 2021, there are 287.28 million shares of the stock outstanding; of those shares, 114.94 million are actually floating and available for transactions.  The ADV for all trading days in 2020 was approximately $V = 10.70$ million shares per day.  Critically for this work, as of January 15, 2021, a total of 44.67 million shares were shorted.  That is, the short interest ratio is given by $s = 44.67/10.70 \approx 4.17$.  Notably this short interest ratio is in excess of $s^* \approx 0.615$.  As such, we find that a strictly positive short squeeze occurs with size $\delta = p_0\left(bsV - \frac{1 - \mu + 2\alpha}{1+\mu}\right) \approx 7.11p_0$, i.e., the price would jump by more than 700\% due to a short squeeze, at an external capital purchase of $C^* \approx \$1.6$ million.  With a price per share of approximately $p_0 = \$2.33$/share prior to the short squeeze, we would anticipate a sudden jump in prices to over \$18.90/share due to the short squeeze; in fact, under such a scenario we compute the clearing price $p^* = \$19.18$/share.  This is a jump in price from approximately \$2.60/share before the short squeeze (due entirely to the external capital purchases $C^*$).  This is consistent with price movements observed on in early 2021 with a significant price spike observed on January 27, 2021. 
\end{example}


\subsection{A speculative case study}\label{sec:cs-rad}
In this section we wish to consider a final, counterfactual, case study.  For this example, we wish to study a stock that, as of the time of this writing, has not experienced a short squeeze but which is heavily short sold.  For this task, we consider Rite Aid Corp.\ (RAD) at the end of the first quarter of 2022 to test what kind of short squeeze event could occur under enough external capital purchases.
\begin{example}\label{ex:rad}[Rite Aid Corp.]
In this final case study, we consider RAD stock at the end of the first quarter 2022.  As of March 31, 2022, there are 55.78 million shares of the stock outstanding; of those shares 52.59 million are actually floating and available for transactions.  The ADV for all trading days in the year prior to March 31, 2022 was approximately $V = 2.13$ million shares per day.  Critically for this work, as of March 31, 2022, a total of 17.22 million shares were shorted.  That is, the short interest ratio is given by $s = 17.22/2.13 \approx 8.08$.  Notably, this short interest ratio is in excess of $s^* \approx 0.615$.  As such, we find that a strictly positive short squeeze occurs with size $\delta = p_0\left(bsV - \frac{1 - \mu + 2\alpha}{1+\mu}\right) \approx 14.93p_0$, i.e., the price would jump by almost 1500\% due to a short squeeze, at an external capital purchase of $C^* \approx \$1.2$ million.  With a price per share of $p_0 = \$8.75$/share at market close on March 31, 2022, we would anticipate a sudden jump in price to approximately \$130.64/share due to the short squeeze; in fact, under such a scenario we compute the clearing price $p^* = \$140.39$/share.  This is a jump in price from approximately \$9.76/share before a possible short squeeze (due entirely to the external capital purchase $C^*$).  At the time of this writing, the author is not aware of any attempted short squeeze on RAD stock with the prior high point coming during the actions of r/WallStreetBets in early 2021 when the price reached approximately \$26.30/share.
Notably, in comparing this example with Examples~\ref{ex:gme} and~\ref{ex:amc}, RAD would require less capital purchases in order to trigger a short squeeze than either GME or AMC did in early 2021; furthermore with the size of the short squeeze at approximately 1493\%, this event would be over double the size of the one experienced by AMC and over 75\% of the size of the short squeeze of GME.
\end{example}

\section{Conclusion}\label{sec:conclusion}
In this work we developed a formulation that provides the equilibrium price due to feedback effects of margin calls on short sales.  In doing so, we found two cutoff levels: (i) $C^*$ external capital that needs to be invested into the asset to trigger a margin call and (ii) $S^*$ shares that, if short sold, will trigger a short squeeze and a rapid jump in prices.  We found, numerically, that recent price movements in several stocks can be attributed to the short squeeze as they have been highly short sold -- above $S^*$ -- and, due to increased retail investment, have likely exceeded $C^*$ in external investments; furthermore, we considered another stock that may be poised for a similar squeeze in the near future.  The threshold to achieve a short squeeze $S^*$ is of wider interest as it provides a threshold for which prices can become unstable.  This threshold would be of particular interest to regulators as the consequences are inherently tied to financial stability as, if $S^*$ is exceeded, short squeezes and price volatility is a natural consequence.

{\footnotesize
\bibliographystyle{plain}
\bibliography{bibtex2}
}

\appendix
\section{Proofs of main results}\label{sec:proofs}
\subsection{Proof of Proposition~\ref{prop:exist}}
\begin{proof}
Let $\Phi(p) := f(\frac{C}{p} + [S - \frac{M}{(1+\mu)p}]^+)$.  Then, by monotonicity of the inverse demand function $f$, it must follow that $\Phi(p) \geq f(C/p) > f(0) = p_0$ and $\Phi(p) \leq f(C/p + S) \leq f(C/p_0+S)$ for any $p \geq p_0$.  Therefore any clearing price must exist in the interval $[p_0,f(C/p_0+S)]$.  As the inverse demand functions under considerations are continuous, existence of a clearing price follows immediately by Brouwer's fixed point theorem.
\end{proof}

\subsection{Proof of Lemma~\ref{lemma:p}}
\begin{proof}
To prove this result, we will first demonstrate that there exists a unique positive price $p^* = f(C/p^*)$ under a no-margin call assumption.  Second, provided the prior proposed price results in a margin call, we will show that there is a unique clearing price $p^* = f(S + [C - \frac{M}{1+\mu}]/p^*)$ that is consistent with a margin call $(1+\mu)Sp^* > M$.
\begin{enumerate}
\item\label{proof:p*1} Under the no-margin call assumption, any clearing price $p^*$ must be a root of $p \mapsto g(p) := f(C/p)-p$.  Immediately, $g'(p) = -\frac{C}{p^2}f'(\frac{C}{p}) - 1 < 0$ for any $p > 0$.  By continuity and strict decreasing-ness of $g$, there exists at most one solution of $g(p) = 0$.  Furthermore, $g(p_0) > 0$ and $g(f(C/p_0)) < 0$ implies that there exists at least one solution of $g(p) = 0$, i.e., there exists a unique positive price assuming no margin call is required.

This price $p^*$ is a clearing price of~\eqref{eq:clearing} if and only if $(1+\mu)Sp^* \leq M = (1+\alpha)Sp_0$.  That is, $p^*$ is a clearing price if $p^* \leq \frac{1+\alpha}{1+\mu}p_0$.  Due to the monotonicity of $g$, this holds if and only if $g(\frac{1+\alpha}{1+\mu}p_0) \leq 0$ or, equivalently, $f(\frac{(1+\mu)C}{(1+\alpha)p_0}) \leq \frac{1+\alpha}{1+\mu}p_0$.  By the monotonicity of the inverse demand function $f$, this holds if and only if
\begin{equation*}
C \leq \frac{(1+\alpha)p_0}{1+\mu}f^{-1}\left(\frac{(1+\alpha)p_0}{1+\mu}\right) =: C^*.
\end{equation*}
\item\label{proof:p*2} Consider now that $C > C^*$, i.e., there is a realized margin call.  Before studying the uniqueness of the clearing price in this setting, we will determine that at least one solution $p^*$ must exist; this follows directly from Proposition~\ref{prop:exist} as we are (implicitly) assuming that no solution exists without a margin call.  In fact, any clearing solution in this setting must be such that $p^* > \frac{1+\alpha}{1+\mu}p_0$ or else no margin call would be required as detailed in the above setting.

Consider now the question of uniqueness.  As in the no-margin call case above, any clearing price must be a root of $p \mapsto g(p) := f(S + [C - \frac{M}{1+\mu}]/p) - p$.  There are three possible cases to consider.  First, if $C > \frac{M}{1+\mu}$ then $g$ is strictly decreasing which, as in the no-margin call case above, immediately implies that there is a unique clearing solution $p^*$.  Second, if $C = \frac{M}{1+\mu}$ then $g(p) = f(S) - p$ and trivially the only possible root is $p^* = f(S)$.  Finally, if $C < \frac{M}{1+\mu}$ then $g$ is strictly concave as 
    \[g''(p) = \frac{C - \frac{M}{1+\mu}}{p^3}\left[2f'\left(S + [C - \frac{M}{1+\mu}]/p\right) + \frac{C - \frac{M}{1+\mu}}{p} f''\left(S + [C - \frac{M}{1+\mu}]/p\right)\right] < 0.\]  
Furthermore, $f(\frac{(1+\mu)C}{(1+\alpha)p_0}) > \frac{1+\alpha}{1+\mu}p_0$ by $C > C^*$, i.e., $g(\frac{(1+\alpha)p_0}{1+\mu}) > 0$.  Thus the first possible root such that $(1+\mu)Sp^* > M$ must also be the only clearing solution due to the concavity of $g$.
\end{enumerate}
\end{proof}

\subsection{Proof of Proposition~\ref{prop:sensitivity}}
\begin{proof}
The form for these partial derivatives follows from a direct application of implicit differentiation (recalling $M = (1+\alpha)Sp_0$) provided $p^* + ([C - \frac{M}{1+\mu}\ind{C \geq C^*}]/p^*)f'(\frac{C}{p^*} + [S - \frac{M}{(1+\mu)p^*}]^+) \neq 0$.  We will prove that this denominator (which is common to both terms) is in fact strictly positive and thus providing also the desired monotonicity (recalling $p^* > \frac{(1+\alpha)p_0}{1+\mu}$ if $C > C^*$).
Note that, by the construction of the clearing prices $p^*$, if $S - \frac{M}{(1+\mu)p^*} > 0$ then it must follow that $C > C^*$.
First, assume $C < C^*$.  Then, by Assumption~\ref{ass:idf}, $p^* + \frac{C}{p^*} f'(\frac{C}{p^*}) > 0$ as desired.
Second, assume $C \geq C^*$ and consider the function $p \mapsto g(p) := f(S + [C - \frac{M}{1+\mu}]/p) - p$ as defined in the proof of Lemma~\ref{lemma:p}\eqref{proof:p*2}.  By that proof, $0 > g'(p^*) = -([C - \frac{M}{1+\mu}]/(p^*)^2)f'(S + [C - \frac{M}{1+\mu}]/p^*) - 1$ at the clearing price $p^*$.  By a simple rearranging of terms, it trivially follows that $p^* + ([C - \frac{M}{1+\mu}]/p^*)f'(S + [C - \frac{M}{1+\mu}]/p^*) > 0$ and the proof is complete.
\end{proof}

\subsection{Proof of Lemma~\ref{lemma:squeeze}}
\begin{proof}
First, there can exist a strictly positive short squeeze $\delta > 0$ if and only if $p \mapsto g(p;C,S) = f(S + \frac{1}{p}[C - \frac{1+\alpha}{1+\mu}Sp_0]) - p$ is strictly increasing at $C = C^*$ and $p = \frac{1+\alpha}{1+\mu}p_0$.  (Otherwise $p^*(C)$ would be continuous at $C^*$.)  Taking the derivative of $g$, we find:
    \[g'\left(\frac{1+\alpha}{1+\mu}p_0;C^*,S\right) = -\frac{1+\mu}{(1+\alpha)p_0} \left[f^{-1}\left(\frac{(1+\alpha)p_0}{1+\mu}\right) - S\right] f'\left(f^{-1}\left(\frac{(1+\alpha)p_0}{1+\mu}\right)\right) - 1.\]
Thus $g'(\frac{1+\alpha}{1+\mu}p_0;C^*,S) > 0$ if and only if $S > S^* := \frac{(1+\alpha)p_0}{(1+\mu) f'\left(f^{-1}\left(\frac{(1+\alpha)p_0}{1+\mu}\right)\right)} + f^{-1}\left(\frac{(1+\alpha)p_0}{1+\mu}\right)$ and the proof is completed.
\end{proof}

\subsection{Proof of Corollary~\ref{cor:sensitivity}}
\begin{proof}
This follows from a direct application of $\frac{\partial}{\partial S}p^*$ as provided in Proposition~\ref{prop:sensitivity}.
\end{proof}

\end{document}